\documentclass{article}

\usepackage{amsmath, amsthm, amssymb}
\usepackage{graphicx}
\usepackage{microtype}
\usepackage{hyperref}

\usepackage[skip=2pt]{caption}
\usepackage[a4paper, portrait, left=1in, right=1in, top=1in, bottom=1in]{geometry}
\usepackage[numbers,sectionbib]{natbib}

\newcommand{\R}{\mathbb{R}}
\newcommand{\E}{\mathbb{E}}
\DeclareMathOperator{\Var}{Var}
\newcommand{\supplementary}[1]{#1}
\newcommand{\short}[1]{}

\newtheorem{theorem}{Theorem}
\newtheorem{corollary}{Corollary}
\newtheorem{lemma}{Lemma}

\newtheorem{innercustomthm}{Restatement of Theorem}
\newenvironment{customthm}[1]
  {\renewcommand\theinnercustomthm{#1}\innercustomthm}
  {\endinnercustomthm}

\date{}

\title{CountSketches, Feature Hashing and the Median of Three\footnote{The authors are part of BARC, Basic Algorithms Research Copenhagen, supported by the VILLUM Foundation grant 16582.}}

\author{
	Kasper Green Larsen\footnote{Kasper Green Larsen is supported by a Villum Young Investigator Grant, a DFF Sapere Aude Research Leader Grant and an AUFF Starting Grant.}\\
	\texttt{larsen@cs.au.dk}\\
	Department of Computer Science,\\
	Aarhus University, Denmark
	\and
    Rasmus Pagh\\
	\texttt{pagh@di.ku.dk}\\
	BARC, Department of Computer Science,\\
	University of Copenhagen, Denmark
	\and
	Jakub Tětek\footnote{Jakub Tětek has been supported by the Bakala Foundation Scholarship.}\\
	\texttt{j.tetek@gmail.com}\\
	BARC, Department of Computer Science,\\
	University of Copenhagen, Denmark
}

\begin{document}

\maketitle

\begin{abstract}
In this paper, we revisit the classic \emph{CountSketch} method, which is a sparse, random projection that transforms a (high-dimensional) Euclidean vector $v$ to a vector of dimension $(2t-1) s$, where $t, s > 0$ are integer parameters. It is known that even for $t=1$, a CountSketch allows estimating coordinates of $v$ with variance bounded by $\|v\|_2^2/s$. For $t > 1$, the estimator takes the median of $2t-1$ independent estimates, and the probability that the estimate is off by more than $2 \|v\|_2/\sqrt{s}$ is exponentially small in $t$. This suggests choosing $t$ to be logarithmic in a desired inverse failure probability. However, implementations of CountSketch often use a small, constant~$t$. Previous work only predicts a constant factor improvement in this setting. 

Our main contribution is a new analysis of Count\-Sketch, showing an improvement in variance to $O(\min\{\|v\|_1^2/s^2,\|v\|_2^2/s\})$ when $t > 1$.
That is, the variance decreases proportionally to $s^{-2}$, asymptotically for large enough $s$. We also study the variance in the setting where an inner product is to be estimated from two CountSketches.
This finding suggests that the \emph{Feature Hashing} method, which is essentially identical to CountSketch but does not make use of the median estimator, can be made more reliable at a small cost in settings where using a median estimator is possible.

We confirm our theoretical findings in experiments and thereby help justify why a small constant number of estimates often suffice in practice. Our improved variance bounds are based on new general theorems about the variance and higher moments of the median of i.i.d.~random variables that may be of independent interest.
\end{abstract} 
\section{Introduction}
CountSketch~\cite{charikar2004finding} is a classic low-memory algorithm for processing a data stream in one pass. It supports estimating the number of occurrences of different data items in the stream, and can also be used for fast inner product estimation, or as a building block for finding heavy hitters (see e.g.~\cite{woodruff2016new}). 
Since its introduction, CountSketch has proved to be a strong primitive for approximate computation on high-dimensional vectors.
Applications in machine learning include feature selection~\cite{aghazadeh2018mission}, neural network compression~\cite{chen2015compressing}, random feature mappings~\cite{pham2013fast}, compressed gradient optimizers~\cite{spring2019compressing}, and multitask learning~\cite{weinberger2009feature} --- see section~\ref{sec:related} for more details.

\subsection{Sketch description}
CountSketch works in the turnstile streaming model, where one is to maintain a sketch of a vector $v \in \R^d$ under updates to the entries. Concretely, the vector $v$ is given in a streaming fashion as a sequence of updates $(i_1,\Delta_1),(i_2,\Delta_2),\dots,$ where an update $(i,\Delta)$ has the effect of setting $v_i \gets v_i + \Delta$ for some $\Delta \in \R$.

The sketch can be stored as a matrix $A$ with $2t-1$ rows and~$s$ columns --- alternatively viewed as a vector of dimension $(2t-1)s$. 
Updates to the sketch are defined by hash functions $h_1,\dots,h_{2t-1}$ and $g_1,\dots,g_{2t-1}$. To initialize an empty CountSketch, we pick a 2-wise independent hash function $h_i : [d] \to [s]$ mapping entries in $v$ to columns of $A$, and a 2-wise independent hash function $g_i : [d] \to \{-1,1\}$ mapping entries in $v$ to a random sign, each for row $i \in [2t - 1]$.\footnote{A $k$-wise independent hash function has independent and uniform random hash values when restricted to any set of up to $k$ keys.}. To process the update $(j,\Delta)$ the update algorithm sets $A_{i,h_i(j)} \gets A_{i,h_i(j)} + g_i(j) \Delta$ for $i=1,\dots,2t-1$. Thus entry $k$ of the $i$th row of $A$ contains the sum of all coordinates $v_j$ such that $h_i(j)=k$, with each such coordinate $v_j$ multiplied by a random sign $g_i(j)$.

\subsection{Frequency estimation}
A frequency estimation query (a.k.a.~point query) asks to return an estimate of an entry $v_j$. CountSketch supports such queries by returning the median of $\{g_i(j) A_{i,h_i(j)} \}_{i=1}^{2t-1}$. The classic analysis of CountSketch shows that for each row $i$ of $A$ and entry $v_j$, the estimate $\hat{v}_j^i = g_i(j)A_{i,h_i(j)}$ has expectation $v_j$ and variance at most $\|v\|_2^2/s$.
Using Chebyshev's inequality, this implies that $\Pr[|\hat{v}_j^i -v_j| \geq 2 \|v\|_2/\sqrt{s}] \leq 1/4$. This is often boosted to a high probability bound by taking the median $\hat{v}_j$ of the $2t-1$ row estimates $\hat{v}_j^1,\dots,\hat{v}_j^{2t-1}$ and using a Chernoff bound to conclude that  $\Pr[|\hat{v}_j -v_j| \geq 2 \|v\|_2/\sqrt{s}] \leq \exp(-\Omega(t))$. A similar, but less common, analysis based on Markov's inequality can also be used to give a bound based on the $\ell_1$ norm of $v$. 
More concretely, it can be shown that $\E[|\hat{v}_j^i - v_j|] \leq \|v\|_1/s$. This can again be combined with the Chernoff bound to conclude that $\Pr[|\hat{v}_j -v_j| \geq 4 \|v\|_1/s] \leq \exp(-\Omega(t))$. This latter bound has a better dependency on the number of columns (and hence space usage) but potentially a worse dependency on $v$ as $\|v\|_1 \geq \|v\|_2$ for all $v$ ($\|v\|_1$ and $\|v\|_2$ are close when $v$ consists of a few large non-zero entries).

Both of the above bounds suggest using a value of $t$ that is logarithmic in the desired failure probability. However, practitioners rarely use more than a small constant number of rows, such as $3$ or $5$ ($t=2,3$) rows. Based on the classic analysis of CountSketch, this only changes the failure probability by a constant factor and has no asymptotic benefits. Nonetheless, we show in experiments (in Section~\ref{sec:exp}) that already $3$ rows seems to have a profound impact on the variance of the estimates. The result of one experiment is seen in Figure~\ref{fig:varianceSynthetic}. Here the ratio between the variance with $1$ and $3$ rows is more than $200$ when using $s=512$ columns.

\begin{figure*}[t]
  \includegraphics[width=\textwidth]{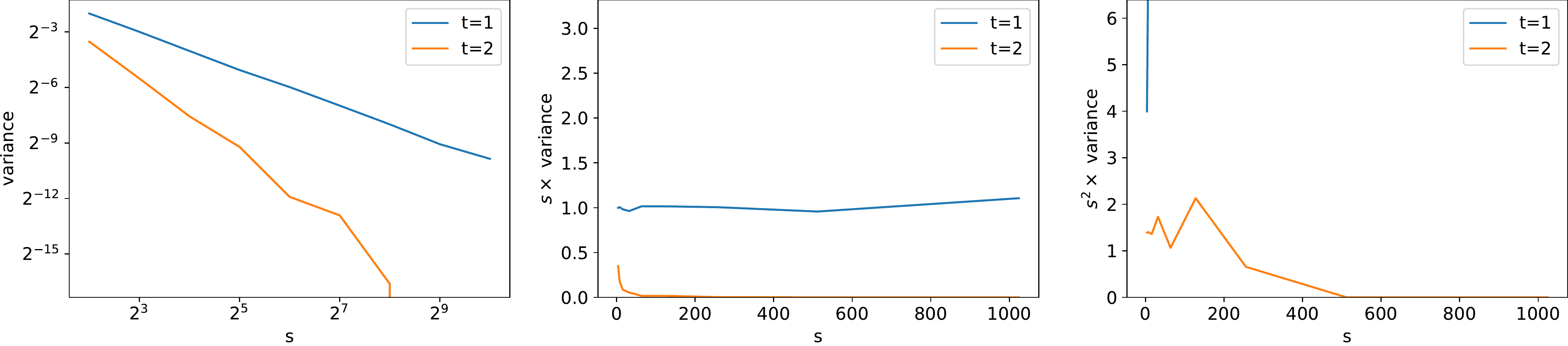}
  \caption{Variance plot of frequency estimation (point queries) for CountSketch with $t=1$ and $t=2$, run on a one-hot vector $v$ with a single nonzero coordinate $v_i=1$. The first figure shows that the variances behave linearly on a log-log plot, suggesting that the variances decrease polynomially with the number of columns $s$. The second plot shows variance multiplied by~$s$. CountSketch with $t=1$ becomes near-constant, suggesting that its variance grows as $1/s$. The third plot shows variance multiplied by $s^2$ and suggests that the variance for $t=2$ grows roughly as $1/s^2$.}
  \label{fig:varianceSynthetic}
\end{figure*}
We explain these observations through new theoretical insights about CountSketch. Concretely, we prove:
\begin{theorem}
  \label{thm:variance}
CountSketch with $t=2$ (3 rows) satisfies $\E[(\hat{v}_j - v_j)^2] \leq \min\{3 \|v\|^2_1/s^2, \|v\|_2^2/s\}$.
\end{theorem}
The new contribution in Theorem~\ref{thm:variance} is the bound in terms of $\|v\|_1$. Quite interestingly, the bound in terms of $\|v\|_1$ is not true if using just a single row. To see this, consider any vector $v$ with a single non-zero entry $v_i$. The estimate for any other entry $v_j$ then equals $0$ with probability $1-1/s$ ($h(i) \neq h(j)$) and it equals $v_ig(i)g(j)$ with probability $1/s$. One therefore has $\E[(\hat{v}_j - v_j)^2] = v_i^2/s = \|v\|^2_1/s$.
This shows that using just three rows instead of a single row effectively reduces the variance of CountSketch by a factor $s$ in terms of $\|v\|_1$. We find this new insight into one of the most fundamental sketching techniques surprising.
We also show that taking the median of three asymptotically reduces the fourth moment of the error in terms of $\|v\|_2$:
\begin{theorem}
  \label{thm:4th}
 CountSketch with $t=2$ (3 rows) satisfies $\E[(\hat{v}_j - v_j)^4] \leq 3 \|v\|_2^4/s^2$.
\end{theorem}

Moreover, we show\short{ in the supplementary material} that this bound is asymptotically optimal.
If we consider the same example as above with a vector $v$ with just a single non-zero entry $v_i$, we again see that when estimating any $v_j$ with $j\neq i$ we have $\E[(\hat{v}_j - v_j)^4] = v_i^4/s = \|v\|_2^4/s$. Thus using $t=2$ ($3$ rows) rather than $t=1$ ($1$ row) reduces the fourth moment by a factor $s$ in terms of $\|v\|_2$. We find it quite remarkable that a constant factor increase in the number of rows increases the utilization of the number of columns by a linear factor both in terms of the variance as a function of $\|v\|_1$ and the fourth moment as a function of $\|v\|_2$. Combined with our experiments, this strongly suggest that one should always use at least $3$ rows in practice.
We extend our results to any $t$ and show:
\begin{theorem}
  \label{thm:2t}
  CountSketch with median of $2t-1$ rows satisfies $\E[|\hat{v}_j - v_j|^t] \leq 2^{2t-1} \|v\|^t_1/s^t$ and $\E[(\hat{v}_j - v_j)^{2t}] \leq 2^{2t-1} \|v\|^{2t}_2/s^{t}$.
\end{theorem}
Thus we can bound the $t$th moment optimally (up to the $2^{2t-1}$ factor) in terms of $\|v\|_1$ and similarly for the $2t$-th moment in terms of $\|v\|_2$.

\subsection{Inner product estimation}
Another use case of CountSketch is in fast inner product estimation. Concretely, given two vectors $v,w \in \R^d$, if one builds a CountSketch on both vectors using \emph{the same} random hash functions $h_1,\dots,h_{2t-1}$ and $g_1,\dots,g_{2t-1}$ (i.e. the same seeds), then one can quickly estimate $\langle v, w \rangle$ from the two sketches. More precisely, let $A^v$ and $A^w$ denote the matrices constructed for $v$ and $w$, respectively. For any row $i$, the inner product $\langle A^v_i, A^w_i\rangle = \sum_{j=1}^s A^v_{i,j} A^w_{i,j}$ is an unbiased estimator of $\langle v,w \rangle$. Moreover, one can show that $\E[(\langle A^v_i, A^w_i\rangle - \langle v, w\rangle)^2] \leq 2 \|v\|_2^2 \|w\|_2^2/s$ if we replace $g$ by a 4-wise independent hash function (rather than just 2-wise). Combining this with Chebyshev's inequality yields
\[\Pr[|\langle A^v_i, A^w_i\rangle - \langle v, w\rangle| > (2 \sqrt{2}) \|v\|_2 \|w\|_2/\sqrt{s}] < 1/4.\]
Finally, as with frequency estimation (point queries), one can take the median over the $2t-1$ row estimates and apply a Chernoff bound to guarantee that the final estimate, denote it $X$, satisfies
\[\Pr[|X - \langle v,w \rangle| >  (2 \sqrt{2}) \|v\|_2 \|w\|_2/\sqrt{s}] < \exp(-\Omega(t)).\]
CountSketch with just a single row, $t=1$, is in fact identical to the popular \emph{feature hashing} scheme~\cite{weinberger2009feature}. Previous work has not shown any asymptotic benefits of taking the median of a small constant number of rows, using e.g.~$t=2$ or $t=3$. Our contribution is new bounds on the variance of such inner product estimates:
\begin{theorem}
  \label{thm:inner}
  For two vectors $v,w \in \R^d$, let $A^v$ and $A^w$ denote the two
matrices representing a CountSketch of the two vectors when using
the same random hash functions, where the $g_i$ are 4-wise independent. 
Let $X$ denote the median of $\langle A^v_i , A^w_i \rangle$ over rows $i=1,\dots,2t-1$.
Then CountSketch with $t=2$ satisfies
$$
\E[(X- \langle v, w\rangle)^2] \leq \min\{3 \|v\|^2_1\|w\|_1^2/s^2, 2\|v\|_2^2\|w\|_2^2/s\},
$$
and for $t > 1$:
\[\E[|X- \langle v, w\rangle|^t] \leq 2^{2t-1} \|v\|_1^t \|w\|_1^t/s^t,\text{ and}\]
\[\E[(X-\langle v, w\rangle)^{2t}] \le 4^{2t-1} \|v\|_2^{2t} \|w\|_2^{2t}/s^{t} \enspace . \] 
\end{theorem}
We note that the bounds in terms of $\|v\|_1^2$ and $\|w\|_1^2$ can be shown only assuming 2-wise independence of the $g_i$.
As with frequency estimation queries, a simple example demonstrates that the variance bound in terms of $\|v\|_1^2 \|w\|_1^2$ is false for $t=1$. Concretely, let $v$ have a single coordinate $v_i$ that is non-zero and let $w$ have a single coordinate $w_j$ with $j \neq i$ that is non-zero. Then $\langle v, w\rangle=0$, yet the probability that $v_i$ and $w_j$ hash to the same entry is $1/s$. In that case, the estimate is either $v_i w_j = \|v\|_1\|w\|_1$ or $-v_i w_j$. This implies that $\E[(X- \langle v, w\rangle)^2] = \|v\|_1^2 \|w\|_1^2/s$, i.e. a factor $s$ worse than the guarantees with three rows.

We have also performed experiments estimating the variance on real-world data sets, see Section~\ref{sec:inner}. When $s$ is large enough (so that $\|v\|_1^2\|w\|_1^2/s^2$ becomes the smallest term), these experiments support our theoretical findings as with the frequency estimation queries.

\paragraph{Discussion.}
Similarly to the frequency estimation queries, our new theoretical bounds and supporting experiments strongly advocates taking the median of at least $3$ rows when using CountSketch for inner product estimation. Equivalently, when using feature hashing for inner product estimation, one should take the median of at least $3$ independent instantiations. This reduces the variance by a linear factor in the number of columns/coordinates of the sketch. We remark that taking the median might not be allowed in all applications. For instance, when using CountSketch/feature hashing as preprocessing for Support Vector Machines, using one row corresponds to a kernel function, while this is not the case when taking the median of multiple row estimates. The median of three can thus not be directly used in this setting.

\subsection{New bounds on moments of the median}
We prove our new variance and moment bounds for CountSketch by showing general theorems relating moments of the median of i.i.d. random variables to smaller moments of the individual random variables. These new bounds are very natural and should have applications besides in CountSketch. Moreover, we show\short{ in the supplementary material} that they are asymptotically optimal.
\begin{theorem}
  \label{thm:moments}
  Let $X_1,\cdots,X_{2t-1}$ be $2t-1$ i.i.d. real valued random variables and let $Y$ denote their median. For all positive integers $q$ it holds that
\[\E[|Y-\E[X_1]|^{t q}] \leq \tbinom{2t-1}{t} \cdot \E[|X_1-\E[X_1]|^q]^{t} \enspace . \]
In particular, $\E[|Y-\E[X_1]|^{t q}] \leq 2^{2t-1} \cdot \E[|X_1-\E[X_1]|^q]^{t}$.
\end{theorem}
In many data science applications, the $X_i$ would be unbiased estimators of some desirable function of a data set, such as e.g. the coordinate $v_i$ in a vector $v$. Theorem~\ref{thm:moments} thus gives a bound on the $tq$-th moment of the estimation error of the median $Y$ in terms of just the $q$-th moment of a single variable. We remark that the median of $2t-1$ unbiased estimators is not necessarily itself an unbiased estimator, thus the bound on $\E[(Y-\E[X_1])^{tq}]$ is much more desirable than a bound on e.g. $\E[(Y-\E[Y])^{tq}]$ as the mean of $Y$ might be tricky to prove an exact bound for. However, one can, in fact, derive a bound on the variance of $Y$ itself (on $\E[(Y-\E[Y])^{2}$) directly from Theorem~\ref{thm:moments}:
\begin{corollary}
  \label{cor:variance}
  Let $X_1,X_2,X_3$ be i.i.d. real valued random variables and let $Y$ denote their median. Then $$\Var(Y)  \leq \E[(Y-\E[X_1])^2] \leq 3 \cdot \E[|X_1-\E[X_1]|]^2 \enspace .$$
\end{corollary}

\begin{proof}
 From Theorem~\ref{thm:moments} with $q=1$ we have
\[\E[(Y-\E[X_1])^2] \leq 3 \cdot \E[|X_1-\E[X_1]|]^2.\]
 Moreover, the minimizing value $\mu$ for the function $\mu \mapsto \E[(Y-\mu)^2]$ is the mean $\mu = \E[Y]$. Therefore we have $\Var(Y) = \E[(Y-\E[Y])^2] \leq \E[(Y-\E[X_1])^2] \leq 3 \cdot \E[|X_1-\E[X_1]|]^2$.
\end{proof}

In this paper, we mainly consider the case $t=2$ with 3 rows --- or equivalently $3$ i.i.d. random variables.

\subsection{Related work}\label{sec:related}

CountSketch was originally proposed in~\cite{charikar2004finding} as a method for finding heavy hitters (i.e.,~frequently occurring elements) in a data stream.
Though there are better methods for finding heavy hitters in insertion-only data streams, CountSketch has the advantage that it is a \emph{linear sketch}, meaning that sketches can be subtracted to form a sketch of the difference of two vectors.
It is known to be space-optimal for the problem of finding approximate $L_p$ heavy hitters in the turnstile streaming model, where both positive and negative frequency updates are possible~\cite{jowhari2011tight}.

\paragraph{Analysis by Minton and Price.}
An improved analysis of the error distribution of CountSketch was given in~\cite{mintonP14improved}, building on work of~\cite{jowhari2011tight}.
The analysis gives non-trivial bounds only when $t$ is a sufficiently large (unspecified) constant, and the exposition focuses on the case $t = \Theta(\log n)$, where $n$ is the dimension of the vector $v$.
Their stated error bounds are incomparable to ours since they are expressed in terms of (residual) $L_2$ norm of $v$.

The reader may wonder if it is possible to derive our results from the analysis in~\cite{mintonP14improved}.
Their error bound for CountSketch based on $\|v_{[\overline{k}]}\|_2$, where $\|v_{[\overline{k}]}\|_2$ is~$v$ with the largest $k$ entries set to $0$. More concretely, it is shown that for a single row of CountSketch, it holds that $\Pr[(\hat{v}_j^i - v_j)^2 > c_0\|v_{[\overline{c_1 s}]}\|^2_2/s] < 1/4$ for some constants $c_0,c_1$. 
The crucial observation is that all entries of $v_{[\overline{c_1 s}]}$ are bounded by $\|v\|_1/(c_1 s)$ and therefore one has $\|v_{[\overline{c_1 s}]}\|^2_2 = O(\|v\|_1\|v\|_1/s)$. Inserting this gives $\Pr[(\hat{v}_j^i - v_j)^2 > c_2 \|v\|^2_1/s^2] < 1/4$ and this may be combined with Chernoff bounds to give high probability bounds for the median of multiple rows in terms of $\|v\|_1$. 
Already with one row, this looks similar to our bound on the variance of the median of $3$ rows (Theorem~\ref{thm:variance}) which stated that $\E[(\hat{v}_j - v_j)^2] \leq 3 \|v\|^2_1/s^2$. However, as our counterexample above suggests, there is no way of extending the ideas of~\cite{mintonP14improved} to prove $\E[(\hat{v}_j^i - v_j)^2] = O(\|v\|^2_1/s^2)$ as it is simply false for $t=1$. Indeed the way~\cite{mintonP14improved} proves their bound is by analysing the $c_1 s$ largest entries separately from the remaining entries, bounding $\E[(\hat{v}_j^i - v_j)^2]$ only for the small entries in $v_{[\overline{c_1 s}]}$. 
Thus our new variance bounds do not follow from their work.

The experiments in~\cite{mintonP14improved} focus on the setting where $t$ is relatively large, with 20 or 50 rows, i.e., about an order of magnitude larger space usage that we have for $t = 2$.

\paragraph{Dimension reduction.}
CountSketch can be used as a \emph{dimensionality reduction} technique that is simpler and more computationally efficient than the classical Johnson-Lindenstrauss embedding~\cite{johnson1984extensions}.
In this setting there is no estimator, the sketch vector is simply considered a vector in $ts$ dimensions.
Generalized versions of CountSketch have been shown to yield a time-accuracy trade-off~\cite{dasgupta2010sparse,kane2014sparser}.

In machine learning, a variant of CountSketch, now known as \emph{feature hashing}, was independently introduced in~\cite{weinberger2009feature}, focusing on applications in multitask learning.
Feature hashing reduces variance in a slightly different way than CountSketch, by initially increasing the dimension of the input vector by a factor $t$ in a way that preserves $L_2$ distances exactly but reduces the $L_\infty$ norm of vectors by a factor $\sqrt{t}$.
In~\cite{chen2015compressing}, CountSketch/feature hashing was wired into the architecture of a neural network in order to reduce the number of model parameters (without the use of medians).
CountSketch has also been used in the construction of \emph{random feature mappings}~\cite{pham2013fast, ahle2020oblivious}, which can be seen as dimension-reduced versions of explicit feature maps.

\paragraph{Further machine learning applications.}
CountSketch, with the median estimator, has been used in several machine learning applications.
In~\cite{aghazadeh2018mission}, CountSketch was used with $t=2$ (3 rows) for large-scale feature selection.
In~\cite{spring2019compressing}, CountSketch was used for compressing gradient optimizers in stochastic gradient descent.
The related \emph{count-min} sketch~\cite{cormode2005improved}, which is the special case of CountSketch where we fix $s(x)=1$, is a popular choice in applications where vectors have non-negative entries.
The count-min estimator takes advantage of non-negativity by taking the minimum of $t$ estimates, and the error distribution can be analyzed in terms of the $L_1$ norm of $v$.
We note that a count-min sketch with a fully random hash function can be used to simulate a CountSketch with $s/2$ entries computing the pairwise difference of entries whose index differ in the last bit (effectively using the least significant bit as the hash function $s$). 
\section{Moments of the Median}
In this section, we prove our new inequalities for moments of the
median. We in fact prove a more general theorem for the median of $2t-1$ i.i.d. random variables. We first state and proof an integral inequality which the proof of the theorem relies on.\allowdisplaybreaks
\begin{lemma} \label{lem:integral_ineq}
Let $f: \mathbb{R}^+ \rightarrow \mathbb{R}^+$ be a non-increasing function and let $t$ be a positive integer. Then
\[
\int_0^\infty f(\sqrt[t]{x})^t dx \leq \Big(\int_0^\infty f(x) dx \Big)^t \enspace .
\]
\end{lemma}
\begin{proof}
Since the function is non-increasing, it is measurable. Moreover, since it is non-negative, the integrals are defined (possibly equal to $+\infty$). 
We have:
\begin{align}
\Big(\int_0^\infty & f(x) dx \Big)^t \nonumber\\
& = \int_0^\infty \dots \int_0^\infty \prod_{i=1}^t f(x_i) dx_1 \dots dx_t \nonumber\\
& = t! \int_0^\infty \int_0^{x_t} \dots \int_0^{x_2} \prod_{i=1}^t f(x_i) dx_1 \dots dx_t \label{eq:one}\\
& \geq t! \int_0^\infty \int_0^{x_t} \dots \int_0^{x_2} f(x_t)^t dx_1 \dots dx_t\label{eq:two}\\
& = t! \int_0^\infty f(x_t)^t \int_0^{x_t} \dots \int_0^{x_2} 1 dx_1 \dots dx_t\nonumber\\
& = t! \int_0^\infty f(x_t)^t \frac{x_t^{t-1}}{(t-1)!} dx_t \label{eq:three}\\
& = \int_0^\infty f(x)^t t x^{t-1} dx = \int_0^\infty f(\sqrt[t]{x})^t dx \enspace . \nonumber
\end{align}

The integral in $(\ref{eq:one})$ is exactly over the set $0 \leq x_1 \leq x_2 \leq \dots \leq x_t$. There are $t!$ such sets, each determined by an ordering of the variables. Since $\prod_{i=1}^t$ is a symmetric function (by comutativity) it integrates to the same value over each of these sets. Moreover, these sets partition the set $[0,\infty)^{t}$ (up to a set of measure 0 corresponding to when two variables are equal). Since we have a partition into $t!$ sets and the integral over each set from the partition is the same, the integral over each set is a $t!$-fraction of the integral over the whole space, and $(\ref{eq:one})$ holds. $(\ref{eq:two})$ holds because $f$ is non-increasing and $x_1\leq x_2, \cdots, x_t$. $(\ref{eq:three})$ holds because the inner integrals correspond to the volume of the $t-1$-dimensional ordered simplex scaled by a factor of $x_t$ and the volume of $t-1$-dimensional ordered simplex is $\frac{1}{(t-1)!}$ (this holds by symmetry, and can be argued the same way as $(\ref{eq:one})$). The final equality holds by substituting $x = x^t$.
\end{proof}
\begin{customthm}{\ref{thm:moments}}
  Let $X_1,\cdots,X_{2t-1}$ be $2t-1$ i.i.d. real valued random variables and let $Y$ denote their median. For all positive integers $q$ it holds that
\[\E[|Y-\E[X_1]|^{t q}] \leq \tbinom{2t-1}{t} \cdot \E[|X_1-\E[X_1]|^q]^{t} \enspace . \]
In particular, $\E[|Y-\E[X_1]|^{t q}] \leq 2^{2t-1} \cdot \E[|X_1-\E[X_1]|^q]^{t}$.
\end{customthm}
\begin{proof}
Notice that since $Y$ is the median of $X_1,\dots,X_{2t-1}$ and the $X_i$'s have the same mean, we can only have $|Y-\E[X_1]|^{t q} \geq x$ when at least $t$ variables $X_i$ have $|X_i-\E[X_i]|^{t q} \geq x$. There are $\binom{2t-1}{t}$ choices for such $t$ variables, so by the union bound, independence and identical distribution of the $X_i$'s, we have for any $x$ that:
\begin{align*}
\Pr[|Y-\E[X_1]|^{t q} \geq x] \leq \tbinom{2t-1}{t} \Pr[|X_1-\E[X_1]|^{t q} \geq x]^t.
\end{align*}
We can thus bound $\E[|Y-\E[X_1]|^{tq}]$ as:
\begin{align*}
\E[|Y&-\E[X_1]|^{t q}] \\
=& \int_0^\infty \Pr[|Y-\E[X_1]|^{t q} \geq x] dx  \\
\leq& \tbinom{2t-1}{t} \int_0^\infty \Pr[|X_1-\E[X_1]|^{t q} \geq x]^t dx \\
=& \tbinom{2t-1}{t} \int_0^\infty \Pr[|X_1-\E[X_1]|^{q} \geq \sqrt[t]{x}]^t dx\\
\leq& \tbinom{2t-1}{t} \Big(\int_0^\infty \Pr[|X_1-\E[X_1]|^{q} \geq x] dx \Big)^t \\
=& \tbinom{2t-1}{t} \cdot \E[|X_1-\E[X_1]|^q]^{t},
\end{align*}
where the first and last equalities hold by a standard identity for non-negative random variables, and the last inequality holds by Lemma \ref{lem:integral_ineq} since $\Pr[|X_1-\E[X_1]|^{q} \geq x]$ is a non-increasing non-negative function.
\end{proof}

\supplementary{
The bound shown in this section can easily be seen to be asymptotically optimal. Consider $X_i$'s which take value $k$ with probability $1/k$ and are zero otherwise. Then 
\begin{align*}
\E[|&Y-\E[X_1]|^{qt}]\\
=& (k-1)^{qt} \Pr[Y = k]\\
\geq & (k-1)^{qt} \tbinom{2t-1}{t} \Pr[X_1=k]^t (1-\Pr[X_1=k])^{t-1}\\
=& \frac{(k-1)^{qt}}{k^t} \tbinom{2t-1}{t} (1-\frac{1}{k})^{t-1}\\
\sim& (k-1)^{(q-1)t} \tbinom{2t-1}{t} 
\end{align*}
where the limit in $\sim$ is taken for $k \rightarrow \infty$. On the other hand, the bound given by our theorem is
\begin{align*}
    \tbinom{2t-1}{t}& \E[|X_1-\E[X_1]|^q]^{t} \\
    =& \tbinom{2t-1}{t} (\frac{1}{k} (k - 1)^q)^t \\
    \sim& (k-1)^{(q-1)t} \tbinom{2t-1}{t} 
\end{align*}
}

\section{CountSketch}
In this section, we prove our new bounds on the variance
(Theorem~\ref{thm:variance}) and 4th moment (Theorem~\ref{thm:4th})
for CountSketch with $3$ rows ($t=2$) as well as our general theorem
with the median of $2t-1$ estimates (Theorem~\ref{thm:moments}). 

\supplementary{
The bounds on frequency estimation are optimal up to $1+o(1)$ factor. This can be seen by considering input consisting of one item with frequency $s$ and a sketch of size $s$. Querying an item with frequency $0$ then reproduces the example from the last section for which our bounds are optimal up to $1+o(1)$ factor.
}

\paragraph{Frequency estimation.}
Recall that CountSketch with three rows computes an estimate
$\hat{v}_j^i$ for each of three rows $i=1,2,3$ and returns the median
$\hat{v}_j$ as its estimate of $v_j$. From Theorem~\ref{thm:moments},
we see that to obtain variance and 4th moment bounds for
$\hat{v}_j$, we only need to bound $\E[|\hat{v}_j^1 -
\E[\hat{v}_j^1]|^q]$ for $q=1,2$. Such bounds essentially follow from
previous work and are as follows:
\begin{lemma}
  \label{lem:count}
CountSketch satisfies $\E[\hat{v}_j^1] = v_j$, $\E[|\hat{v}_j^1 -
v_j|] \leq \|v\|_1/s$ and $\E[(\hat{v}_j^1-v_j)^2] \leq \|v\|_2^2/s$.
\end{lemma}

Theorem~\ref{thm:variance} follows by instantiating
Theorem~\ref{thm:moments} with $q=1$ and the facts $\E[\hat{v}_j^1] = v_j$, $\E[|\hat{v}_j^1 -
v_j|] \leq \|v\|_1/s$ from
Lemma~\ref{lem:count}. Theorem~\ref{thm:4th} follows by instantiating
Theorem~\ref{thm:moments} with $q=2$ and the facts $\E[\hat{v}_j^1] =
v_j$, $\E[(\hat{v}_j^1-v_j)^2] \leq \|v\|_2^2/s$ from
Lemma~\ref{lem:count}. Finally, Theorem~\ref{thm:2t} also follows as
an immediate corollary of Theorem~\ref{thm:moments} and
Lemma~\ref{lem:count}.

\supplementary{
We give the proof of Lemma~\ref{lem:count} in the following for
completeness:
\begin{proof}[Proof of Lemma~\ref{lem:count}]
  For short, let $X = \hat{v}_j^1$, $g = g_1$ and $h=h_1$. We then have:
  \begin{eqnarray*}
    \E[X] &=& \E[g(j)A_{1,h(j)}] \\
          &=& \E[g(j)(g(j)v_j + \sum_{i \neq j} 1_{h(i)=h(j)} g(i) v_i)] \\
    &=& v_j + \sum_{i \neq j} \E[1_{h(i)=h(j)} g(j) g(i)] v_i.
  \end{eqnarray*}
  By independence of $h$ and $g$, and $g$ being 2-wise independent,
  we have
  $$
  \E[1_{h(i)=h(j)} g(j) g(i)] = \E[1_{h(i)=h(j)}
  ]\E[g(j)]\E[g(j)] = 0$$
  and we conclude $\E[X] = v_j$. Next consider
  \begin{eqnarray*}
    \E[|X - \E[X]|] &=& \E[|g(j)A_{1,h(j)} - v_j|] \\
                    &=& \E[|\sum_{i \neq j} 1_{h(i)=h(j)} g(j) g(i) v_i|] \\
                    &\leq& \E[\sum_{i \neq j} |1_{h(i)=h(j)} ||g(j) g(i)||v_i|] \\
                    &=& \sum_{i \neq j} \E[1_{h(i)=h(j)}] |v_i| \\
                    &\leq& \sum_i |v_i|/s \\
    &=& \|v\|_1/s.
  \end{eqnarray*}
  Above we used 2-wise independence of $h$ when we concluded that
  $\E[1_{h(i)=h(j)}] = \Pr[h(i)=h(j)] \leq 1/s$ for all $i \neq j$.
Finally consider:
  \begin{align*}
    & \E[(X - \E[X])^2]\\
    &= \E\left[ \left(\sum_{i \neq j} 1_{h(i)=h(j)} g(j)
                          g(i) v_i \right)^2\right] \\
    &= \E\left[ \sum_{i \neq j} \sum_{k \neq j} 1_{h(i)=h(j)} 1_{h(k)=h(j)} g(j)^2
        g(i)g(k) v_iv_k \right] \\
    &= \sum_{i \neq j} \sum_{k \neq j} \E[1_{h(i)=h(j)} 1_{h(k)=h(j)}]\E[ 
        g(i)g(k)]v_iv_k.
  \end{align*}
  Here we notice by 2-wise independence of $g$ that $\E[g(i)g(k)] = 0$
  whenever $i \neq k$ and $1$ otherwise. The above is thus bounded by:
  \begin{eqnarray*}
    \E[(X - \E[X])^2] &\leq& \sum_{i \neq j} \E[1_{h(i)=h(j)}^2] v_i^2
    \\
                      &\leq& \sum_{i \neq j} v_i^2/s \\
    &\leq& \|v\|_2^2/s.
  \end{eqnarray*}
\end{proof}
}

\paragraph{Inner product estimation.}
Similarly to the case of frequency estimation (point queries), we
prove our new guarantees in Theorem~\ref{thm:inner} by invoking our
general theorems on moments of the median. All we need is moment
bounds for a single row. The following is more or less standard. \short{We include the proof in the supplementary material.}
\supplementary{We show the following (which are more or less standard):}
\begin{lemma}
  \label{lem:innerest}
For two vectors $v,w \in \R^d$, let $A^v$ and $A^w$ denote the two
matrices representing a CountSketch of the two vectors when using
the same random hash functions. Then $\E[\langle A^v_1, A^w_1 \rangle]
= \langle v, w\rangle$ and $\E[|\langle A^v_1, A^w_1 \rangle -
\langle v ,w \rangle |] \leq \|v\|_1\|w\|_1/s$. Moreover, if $g$ is
4-wise independent, then we also have $\E[(\langle A^v_1, A^w_1 \rangle -
\langle v ,w \rangle )^2] \leq 2\|v\|_2^2\|w\|_2^2/s$.
\end{lemma}
Theorem~\ref{thm:inner} follows by combining
Lemma~\ref{lem:innerest} and Theorem~\ref{thm:moments}.\supplementary{\begin{proof}
  For short, let $g = g_1$ and $h = h_1$. We start by observing the
  \begin{eqnarray*}
    \langle A^v_1, A^w_1 \rangle &=& \sum_{i=1}^d \sum_{j=1}^d g(i)
                                     v_ig(j) w_j 1_{h(i)=h(j)}.
  \end{eqnarray*}
  Using 2-wise independence of $g$, we get:
  \begin{eqnarray*}
    \E[\langle A^v_1, A^w_1  \rangle ]&=& \sum_{i=1}^d \sum_{j=1}^d \E[g(i)
                                         v_ig(j) w_j 1_{h(i)=h(j)}] \\
    &=& \sum_{i=1}^d \sum_{j=1}^d \E[g(i)g(j)]
        v_iw_j \E[1_{h(i)=h(j)}] \\
                                     &=& \sum_{i=1}^d v_i w_i \\
    &=& \langle v,w\rangle.
  \end{eqnarray*}
  Next, we see that
  \begin{align*}
    & \E[|\langle A^v_1, A^w_1 \rangle -
\langle v ,w \rangle |]\\
   &= \E\left[\left| \sum_{i=1}^d \sum_{j \neq i} g(i)
                                         v_ig(j) w_j
                            1_{h(i)=h(j)}\right|\right] \\
    &\leq \E\left[\sum_{i=1}^d \sum_{j \neq i} |g(i)|
                                         |v_i||g(j)||w_j||
           1_{h(i)=h(j)}|\right] \\
    &= \sum_{i=1}^d \sum_{j \neq i}
                                         |v_i||w_j|
        \E[1_{h(i)=h(j)}] \\
                        &\leq 
                          \sum_{i=1}^d |v_i|\sum_{j=1}^d
                               |w_j|/s \\
    &= \|v\|_1 \|w\|_1/s.
  \end{align*}
  And finally for 4-wise independent $g$, we have:
  \begin{align*}
    & \E[(\langle A^v_1, A^w_1 \rangle -
\langle v ,w \rangle)^2]\\
   &= \E \left[ \left(\sum_{i=1}^d \sum_{j \neq i} g(i)
                                         v_ig(j) w_j
                             1_{h(i)=h(j)}\right)^2\right] \\
    &= \sum_{i=1}^d \sum_{j \neq i}\sum_{a=1}^d \sum_{b \neq a} \E[g(i)
                                         v_ig(j) w_j
                             1_{h(i)=h(j)} g(a) v_a g(b) w_b
        1_{h(a)=h(b)}] \\
                         &= \sum_{i=1}^d \sum_{j \neq i}\sum_{a=1}^d \sum_{b \neq a} 
                                         v_iw_j
                             \E[1_{h(i)=h(j)} 1_{h(a)=h(b)} ] \E[g(i)g(j)g(a)g(b)] v_a w_b.
\end{align*}
Recall that $a \neq b$ and $i \neq j$. Thus if $a \notin \{i,j\}$
or $b \notin \{i,j\}$, then at least one $g(\cdot )$ is
independent of the remaining three by 4-wise independence of $g$. The
expectation $\E[g(i)g(j)g(a)g(b)]$ then splits into the product of the expectation of that
single term and the remaining three. Since $\E[g(\cdot)]=0$, the whole
term in the sum becomes $0$. Thus for any given $(i,j)$ with $i \neq j$, there
are two choices of $(a,b)$ that do not result in the term
disappearing, namely $(a,b) = (i,j)$ and $(a,b) = (j,i)$. In both
these cases, $g(i)g(j)g(a)g(b)=1$. When $(a,b)=(i,j)$ we have $v_i w_j
v_a w_b = v_i^2 w_j^2$ and when $(a,b)=(j,i)$ we have $ v_i w_j
v_a w_b = v_i w_i v_j w_j$. Therefore:
\begin{align*}
    & \E[(\langle A^v_1, A^w_1 \rangle -
\langle v ,w \rangle)^2]\\
&= \sum_{i=1}^d \sum_{j \neq i}
                                         (v_i^2w_j^2 + v_i w_i v_j
                             w_j)\E[1_{h(i)=h(j)}] \\
  &\leq  \sum_{i=1}^d \sum_{j=1}^d
                                         (v_i^2w_j^2 + v_i w_i v_j
         w_j)/s \\
  &=  \|v\|_2^2 \|w\|_2^2/s + \sum_{i=1}^d \sum_{j=1}^d v_i w_i v_j
      w_j/s \\
  &=  \|v\|_2^2 \|w\|_2^2/s + \langle v, w\rangle^2/s.
\end{align*}
By Cauchy-Schwartz, we have $\langle v, w \rangle \leq \|v\|_2
\|w\|_2$ and thus the whole thing is bounded by $2 \|v\|_2^2 \|w\|_2^2/s$.
\end{proof}} 
\section{Experiments}
\label{sec:exp}
In this section, we empirically support our new theoretical bounds by estimating the variance of CountSketch with $1$ row and $3$ rows on different data sets. We implemented CountSketch in C++ using the multiply-shift hash function~\cite{dietzfelbinger96} as the 2-wise independent hash functions $h$ and $g$. 
We seeded the hash functions using random numbers generated using the built-in Mersenne twister 64-bit pseudorandom generator. 
Experiments were run both for frequency estimation (Section~\ref{sec:freq}) and for inner product estimation (Section~\ref{sec:inner}).

\paragraph{Frequency estimation.}
\label{sec:freq}
We ran experiments on two real-world data sets and two synthetic data sets. The real-world data sets come in the form of a stream of items, with the same item occurring multiple times. Instead of running numerous $(i,1)$ updates ($v_i \gets v_i + 1$), we have simply computed the number of occurrences $c_i$ of each item. We then normalize the occurrences $c_i \gets c_i /\sum_j c_j$ to obtain unit $\ell_1$-norm and then run a single update $v_i \gets v_i + c_i$ for each item $i$ at the end. This produces the exact same CountSketch as when processing the updates one by one (with normalization). The data sets are described in the following:

\begin{itemize}
\item \textbf{Kosarak:} An anonymized click-stream dataset of a Hungarian online news portal.~\footnote{Provided by Ferenc Bodon to the FIMI data set located at \url{http://fimi.uantwerpen.be/data/}.} It consists of transactions, each of which has several items. We created a vector with one entry for each item, storing the total number of occurrences of that item. The vector has $41270$ entries, and when normalized to have $\ell_1$-norm $1$, its $\ell_2$-norm is $0.112$ and the largest entry is $0.075$.

\item \textbf{Sentiment140:} A collection of 1.6M tweets from Twitter~\cite{go2009twitter}. We extracted all words that occur at least twice, and created a vector with one entry per word, containing the total number of occurrences of that word in the tweets. The vector has $147071$ entries, and when normalized to have $\ell_1$-norm $1$, its $\ell_2$-norm is $0.0773$ and the largest entry is $0.0382$.

  \item \textbf{Zipfian:} The Zipfian distribution with skew $\alpha$ and $n$ items is a probability distribution where the $k$th item has probability $k^{-\alpha}/\sum_{j=1}^n j^{\alpha}$. 
Such distributions have been shown to fit a large variety of real-world data.
We created two data sets with $n = 1000$ items using skews $\alpha=0.8$ and $\alpha=1.2$, considering the vector of probabilities. For $\alpha=0.8$, the $\ell_2$-norm is $0.097$ and the largest entry is $0.065$. For $\alpha=1.2$, the $\ell_2$-norm is $0.2713$ and the largest entry is $0.231$. \short{We include results for $\alpha=0.8$ in the supplementary material.}
\end{itemize}

The results of the experiments can be seen in Figures~\ref{fig:varianceKosarak}-\ref{fig:varianceZipfian12}. For each experiment, we plot the variance as a function of the number of columns $s$. We run experiments with $s=2^2,2^3,\dots,2^{10}$ on each data set. For each choice of $s$, we estimate the variance by constructing 1000 CountSketches on the input with new randomness for each. For each CountSketch we pick 100 random items and compute the estimation error for each. We sum the squares of all these estimation errors and divide by $100 \times 1000$ (for small data sets with less than $5000$ items, we instead build $10^6$ CountSketches and make a single estimation on each).

\begin{figure*}[h!t]
  \centering
  \includegraphics[width=0.93\textwidth]{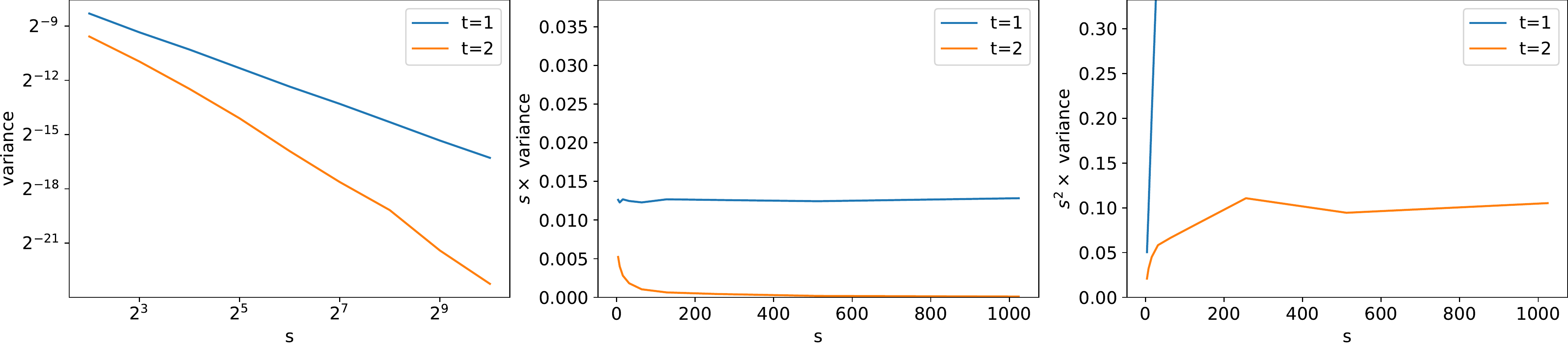}
  \caption{Variance experiments on the Kosarak data set.}
  \label{fig:varianceKosarak}
\end{figure*}

\begin{figure*}[h!t]
  \centering
  \includegraphics[width=0.93\textwidth]{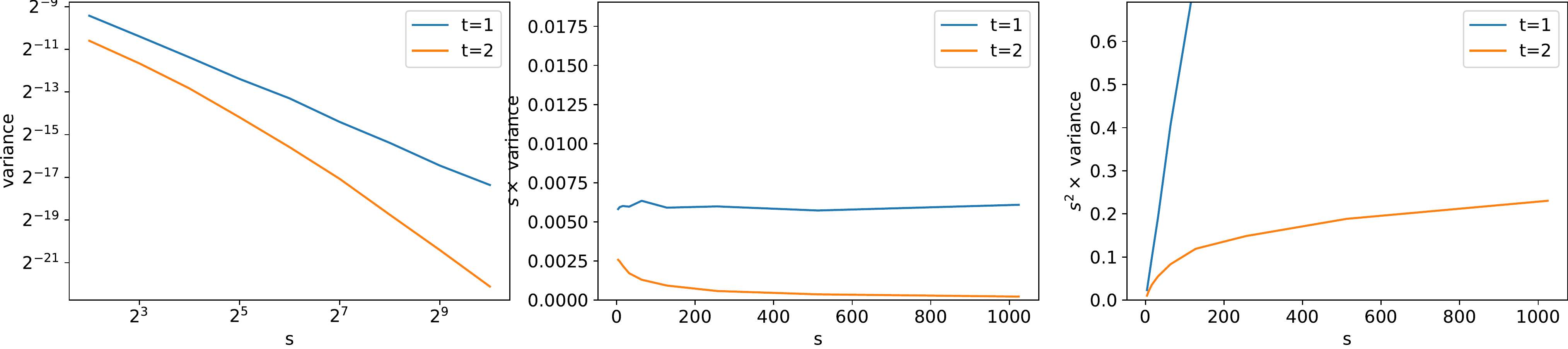}
  \caption{Variance experiments on the Sentiment140 data set.}
  \label{fig:varianceSentiment140}
\end{figure*}

\supplementary{
\begin{figure*}[h!t]
  \centering
  \includegraphics[width=0.93\textwidth]{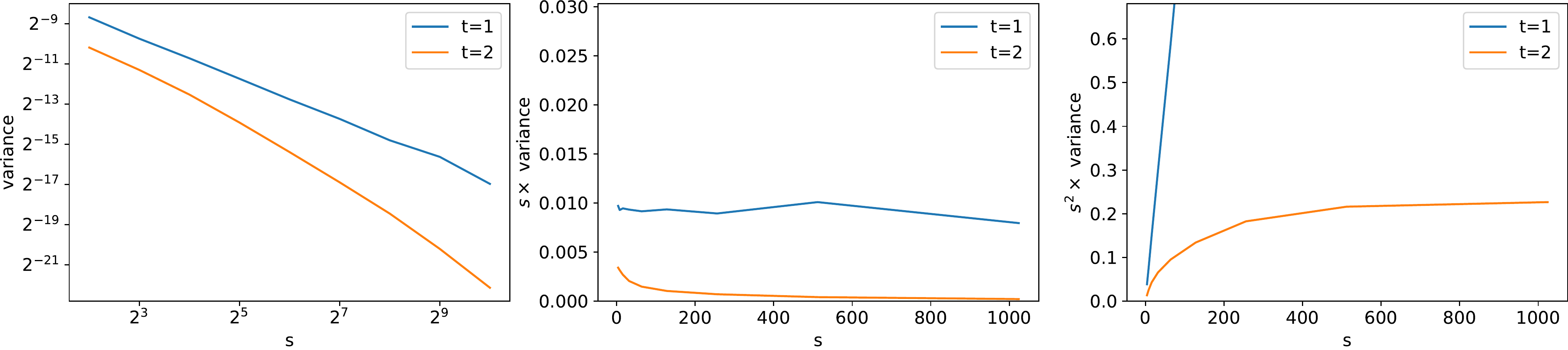}
  \caption{Variance experiments on Zipfian distribution with skew $\alpha=0.8$.}
  \label{fig:varianceZipfian08}
\end{figure*}
}

\begin{figure*}[h!t]
  \centering
  \includegraphics[width=0.93\textwidth]{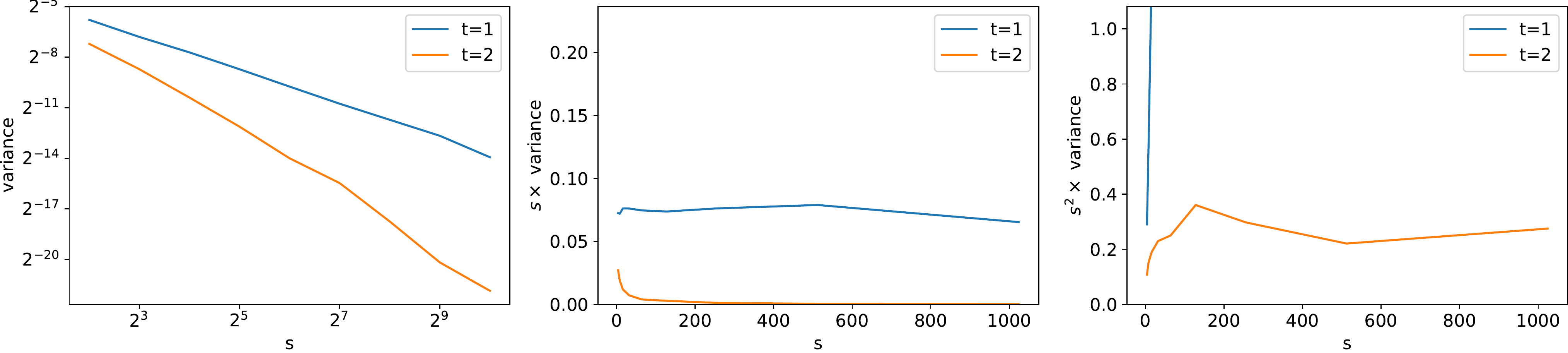}
  \caption{Variance experiments on Zipfian distribution with skew $\alpha=1.2$.}
  \label{fig:varianceZipfian12}
\end{figure*}

On all four data sets, we make three plots of the data. On the first, we show a log-log plot and observe that in all experiments, the variances look linear on the plot, supporting a polynomial dependency on $s$. Second, we scale the variances by $s$ and plot it on a linear scale. In all experiments, the scaled variance for $t=1$ looks constant, supporting a $1/s$ dependency on the number of columns $s$. Third, we scale the variance by $s^2$ and plot it on a linear scale. The scaled variance for $t=2$ looks almost constant in all experiments, supporting a $1/s^2$ dependency on the number of columns. We remark that our theoretical bound in Theorem~\ref{thm:variance} guarantees $\E[(\hat{v}_j-v_j)^2] \leq 3 \|v\|_1^2/s^2$. Since $\|v\|_1 = 1$ in all our data sets, so we expect a CountSketch with $t=2$ on the third plots to stay below $3$ on the y-axis, which it does in all experiments (it even stays below $0.4$).

\begin{table}[h!]
  \centering
  {\small
  \begin{tabular}{|l|c|c|c|}
    \hline
    Data Set & Variance $t=1$ & Variance $t=2$ & Ratio \\
    \hline
    Kosarak & $1.25 \times 10^{-5}$ & $1.42 \times 10^{-7}$ & $88.0$ \\
    Sentiment140 & $5.94 \times 10^{-6}$ & $2.13 \times 10^{-7}$ & $27.9$ \\
    Zipfian $\alpha=0.8$ & $9.56 \times 10^{-6}$ & $2.09 \times 10^{-7}$ & $45.7$ \\
    Zipfian $\alpha=1.2$ & $6.94 \times 10^{-5}$ & $3.99 \times 10^{-7}$ & $173.9$ \\
    \hline
  \end{tabular}
  }
  \caption{Variances for different data sets with $2$ and $3$ rows ($t=1,2$) of CountSketch. In all experiments, we consider a CountSketch with $s=1024$ columns. The ratio in the last column of the table gives the relative difference between using $1$ and $3$ rows.}
    \label{tbl:vars}
  \end{table}

  Table~\ref{tbl:vars} shows the variance on the different data sets using CountSketch with $s=1024$ rows. In all cases, that increasing CountSketch parameter $t$ from $1$ to $2$ clearly provides major reductions in variance, ranging from a factor about $28$ to $174$. 

  We also perform experiments measuring the $4$th moment of the estimation errors. \short{These can be found in the supplementary material.}
\supplementary{
  The results of these experiments can be seen in Figures~\ref{fig:4thKosarak}-\ref{fig:4thZipfian12}.
 
\begin{figure*}[h!t]
  \centering
  \includegraphics[width=0.93\textwidth]{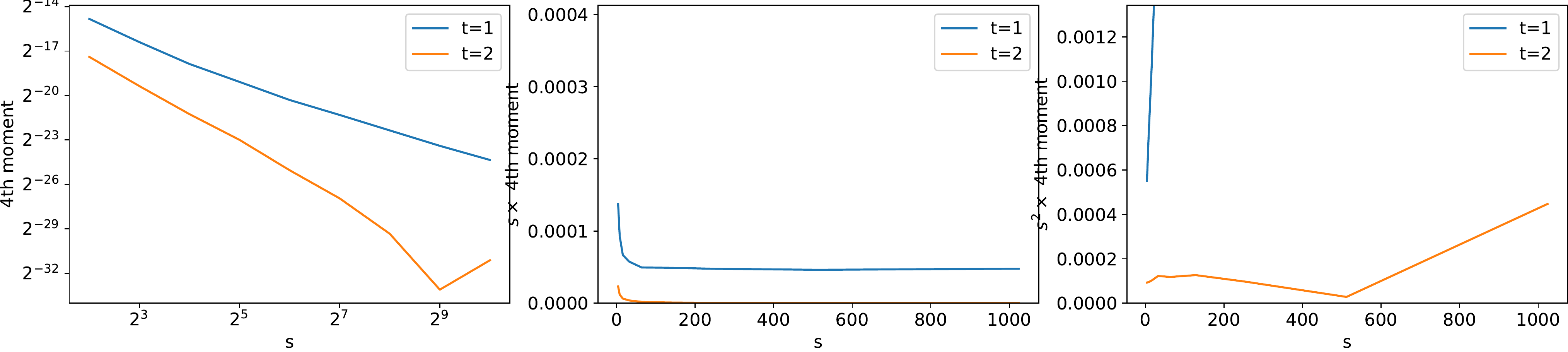}
  \caption{$4$th moment  experiments on the Kosarak data set.}
  \label{fig:4thKosarak}
\end{figure*}

\begin{figure*}[h!t]
  \centering
  \includegraphics[width=0.93\textwidth]{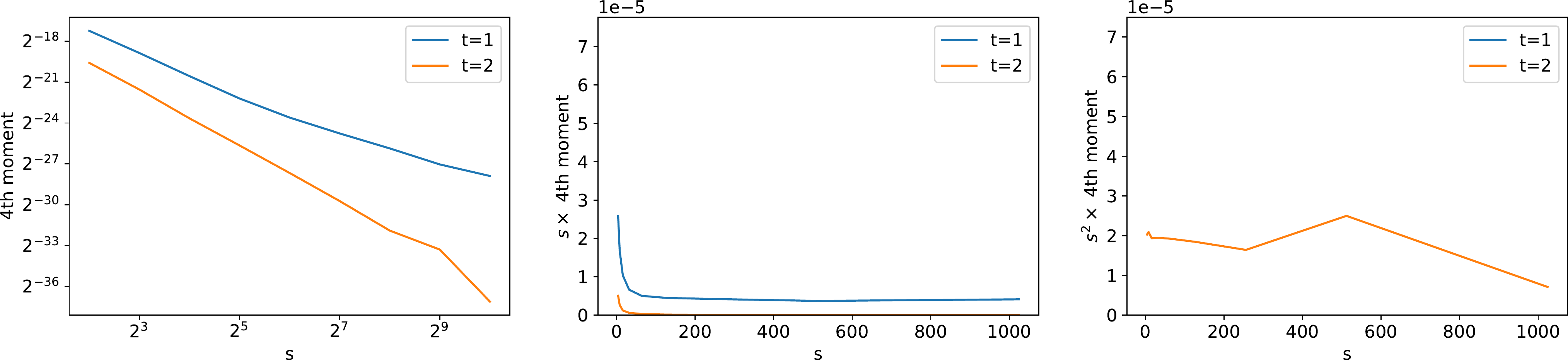}
  \caption{$4$th moment  experiments on the Sentiment140 data set.}
  \label{fig:4thSentiment140}
\end{figure*}

\begin{figure*}[h!t]
  \centering
  \includegraphics[width=0.93\textwidth]{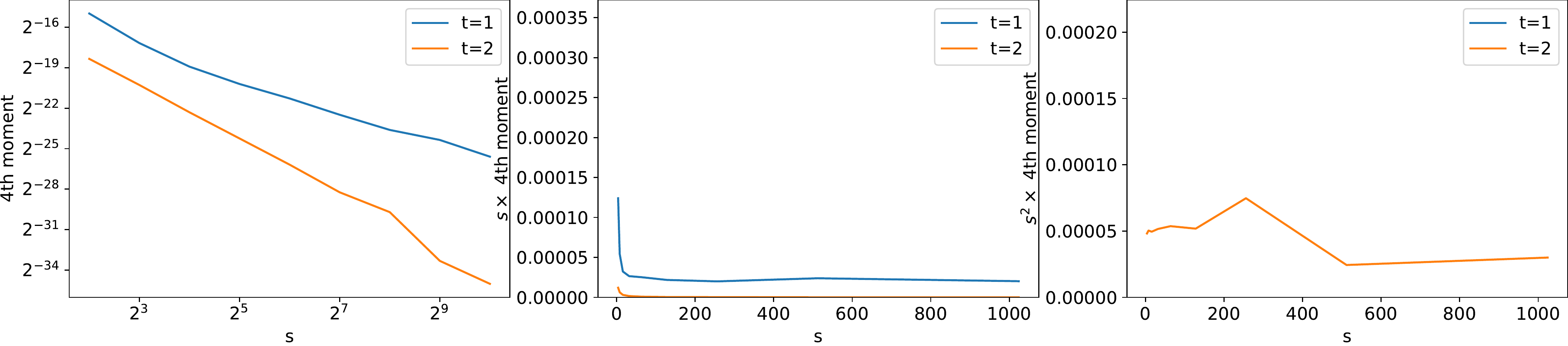}
  \caption{$4$th moment  experiments on Zipfian distribution with skew $\alpha=0.8$.}
  \label{fig:4thZipfian08}
\end{figure*}

\begin{figure*}[h!t]
  \centering
  \includegraphics[width=0.93\textwidth]{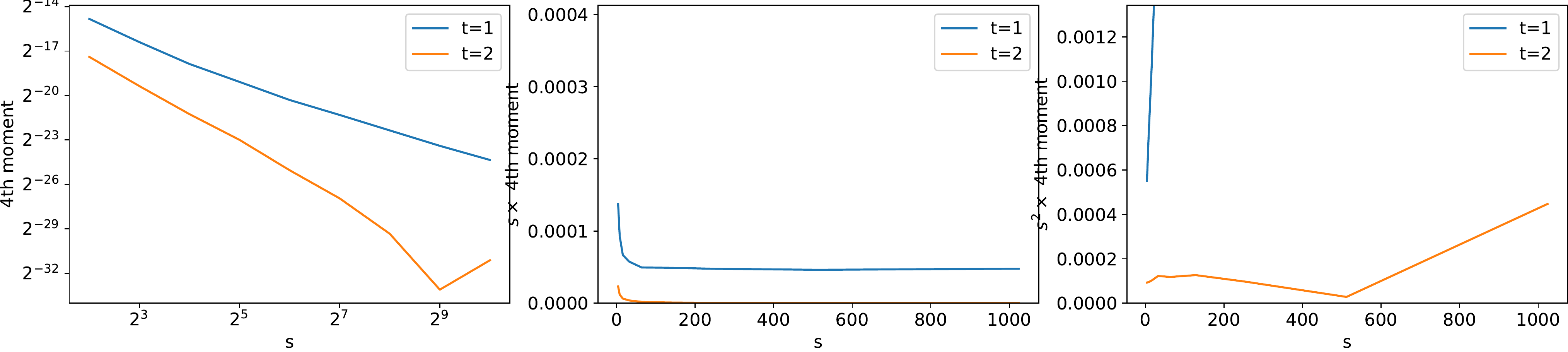}
  \caption{$4$th moment  experiments on Zipfian distribution with skew $\alpha=1.2$.}
  \label{fig:4thZipfian12}
\end{figure*}

Again, we have plotted the $4$th moment times $s$ and the $4$th moment times $s^2$. Similar to the variance experiments, it appears that CountSketch with $t=1$ has a $4$th moment error growing as $1/s$ and with $t=2$ it grows as $1/s^2$, supporting our new theoretical findings in Theorem~\ref{thm:4th}.
}

To summarize, we believe our empirical findings support our new theoretical bounds on the variance and $4$th moment. Moreover, our results strongly suggest that practitioners use $t\geq 2$ with CountSketch as it provides major reductions in variance at little increase in time and memory efficiency.

\paragraph{Inner product estimation.}
\label{sec:inner}
In the following, we perform experiments where we use CountSketch for inner product estimation. We perform experiments on two data sets, a synthetic and a real-world data set:
\begin{itemize}
  \item \textbf{Disjoint 64 non-zeros:} A synthetic data set with two vectors both having  $64$ non-zero entries each with value $1/64$. The two vectors have disjoint supports and thus inner product $0$. The $\ell_2$-norm of the vectors is $1/8=0.125$ and the largest entry is $1/64 \approx 0.0156$.
  \item \textbf{News20:} A collection of newsgroup documents on different topics~\footnote{ \url{http://qwone.com/~jason/20Newsgroups/}}. Each document is represented by a tf-idf vector constructed on the words occurring in the documents. We used the training part of the data set for our experiments. The data set has 11314 distinct vectors. For comparison to our theoretical bounds, we normalize the vector $v$ representing each document such that it has $\|v\|_1=1$. After normalization, the average $\ell_2$-norm of a document vector is $0.1235$ and the average largest entry is $0.0498$.
  \end{itemize}
  \begin{figure*}[ht!]
  \centering
  \includegraphics[width=0.93\textwidth]{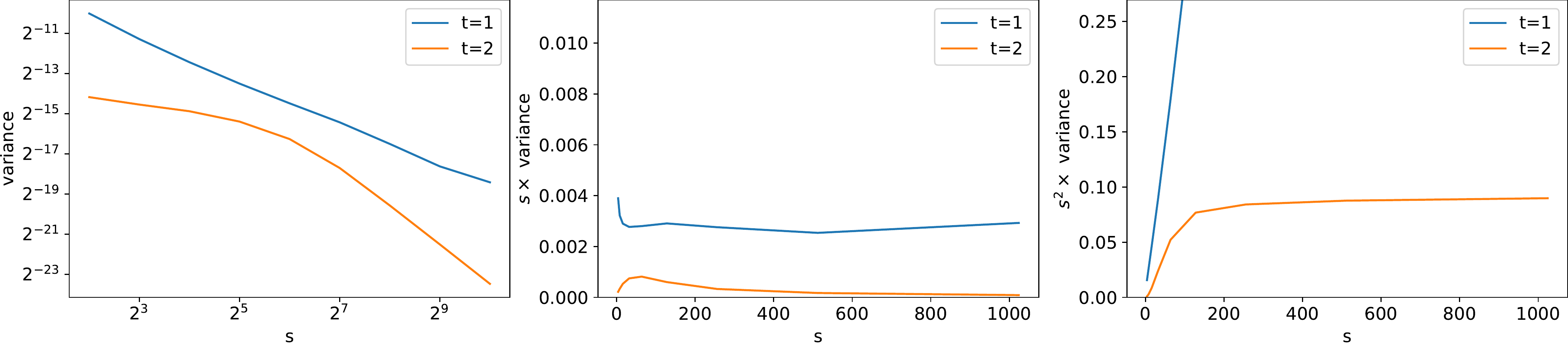}
  \caption{Variance experiments on the Disjoint 64 Non-Zeros data set.}
  \label{fig:var64Disj}
\end{figure*}
\begin{figure*}[h!t]
  \centering
  \includegraphics[width=0.93\textwidth]{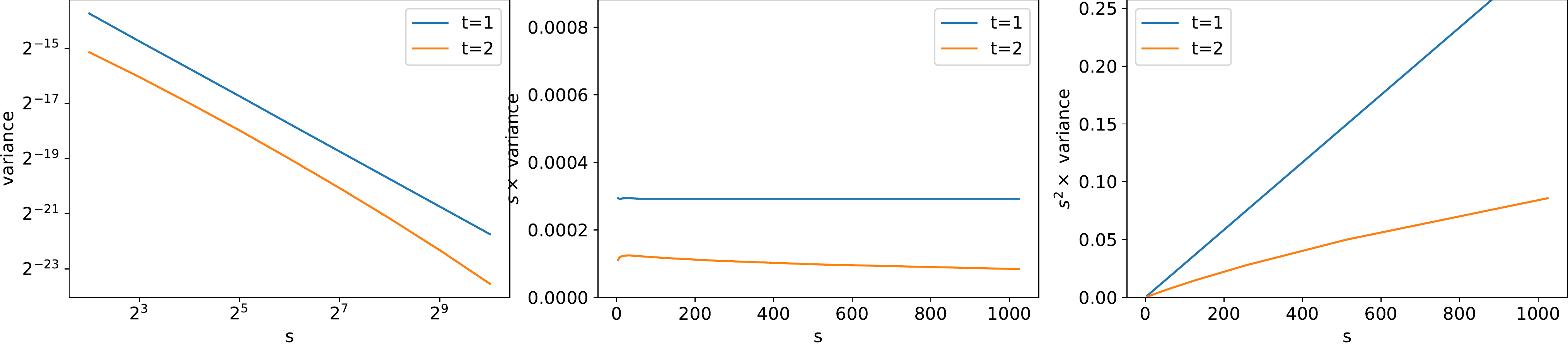}
  \caption{Variance experiments on the News20 data set.}\label{fig:varNews20}
\end{figure*}
\begin{figure*}[h!]
  \centering
  \includegraphics[width=0.93\textwidth]{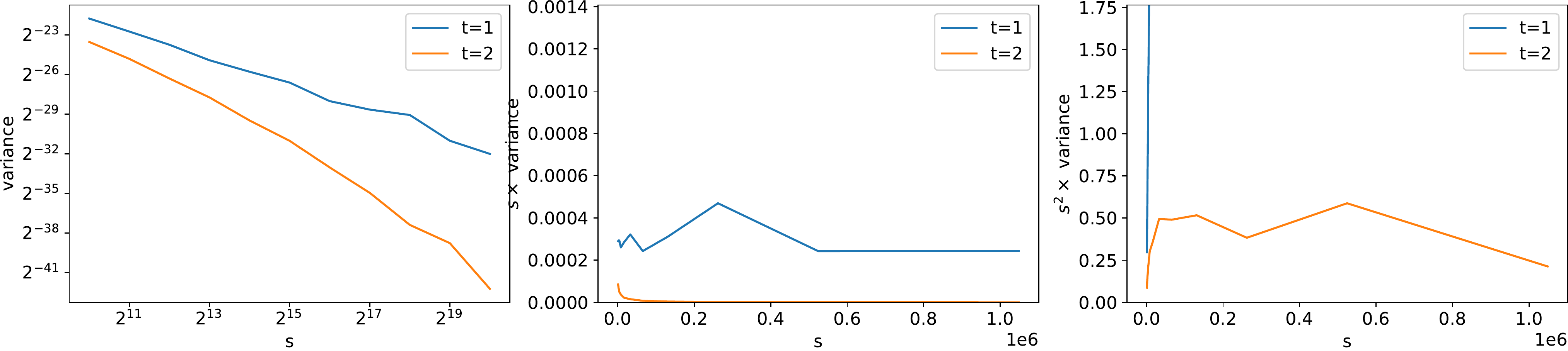}
  \caption{Variance experiments on the News20 data set and number of columns up to $s = 2^{20}$.}
  \label{fig:varNews20LargeS}
\end{figure*}
  For the Disjoint 64 Non-Zeros data set, for $10^6$ iterations, we constructed a new CountSketch on the two vectors using the same random hash functions. We then computed the squared error of the estimates and averaged over all $10^6$ iterations. For the News20 data set, we run $1000$ iterations where we pick new random hash functions in each iteration. In an iteration, we pick $100$ random pairs of distinct vectors, build a CountSketch on both vectors in a pair, and compute the squared estimation error. We finally average over all $100 \times 1000$ pairs.
Figure~\ref{fig:var64Disj} shows the results of experiments on the Disjoint 64 Non-Zeros data set. As before, these plots fit our theoretical guarantees in Theorem~\ref{thm:inner}.

Finally, we have run experiments on the News20 data set. The results are shown in Figure~\ref{fig:varNews20}. 
Unlike in previous experiments, it appears that CountSketch with $3$ rows ($t=2$) has a variance decreasing as $1/s$, not $1/s^2$. To explain this, recall that the guarantee from Theorem~\ref{thm:inner} is $\E[(X - \langle v, w\rangle)^2] \leq \min\{3 \|v\|_1^2\|w\|_q^2/s^2 ,2\|v\|_2^2 \|w\|_2^2\}$. In the News20 data set, the average $\|v\|_2$ is $0.1235$. When this is raised to the fourth power (it appears in both $\|v\|_2^2$ and $\|w\|_2^2$) it becomes very small compared to $\|v\|_1^2\|w\|_1^2 = 1$, thus the $1/s^2$ dependency should only kick in for large values of $s$. To confirm this, we have run more experiments, this time with values of $s$ ranging from $2^{10}$ to $2^{20}$. The results are shown in Figure~\ref{fig:varNews20LargeS}.

With these larger values of $s$, we see the expected $1/s^2$ dependency in the variance for $t=2$. To conclude on this, one may need a larger value of $s$ to see the $1/s^2$ behaviour in variance when performing inner product estimation compared to frequency estimation. This is due to the dependency on the \emph{product} of two vectors of either $\|v\|_1^2 \|w\|_1^2$ or $\|v\|_2^2 \|w\|_2^2$ compared to just the single dependency on $\|v\|_1^2$ and $\|v\|_2^2$ for frequency estimation.

As with frequency estimation, we also experimentally examine the 4th moments. \short{These results are included in the supplementary material.}\supplementary{For the results of these experiments, see Figures~\ref{fig:4th64Disj} and~\ref{fig:4thNews20}.

\begin{figure*}[ht!]
  \centering
  \includegraphics[width=0.93\textwidth]{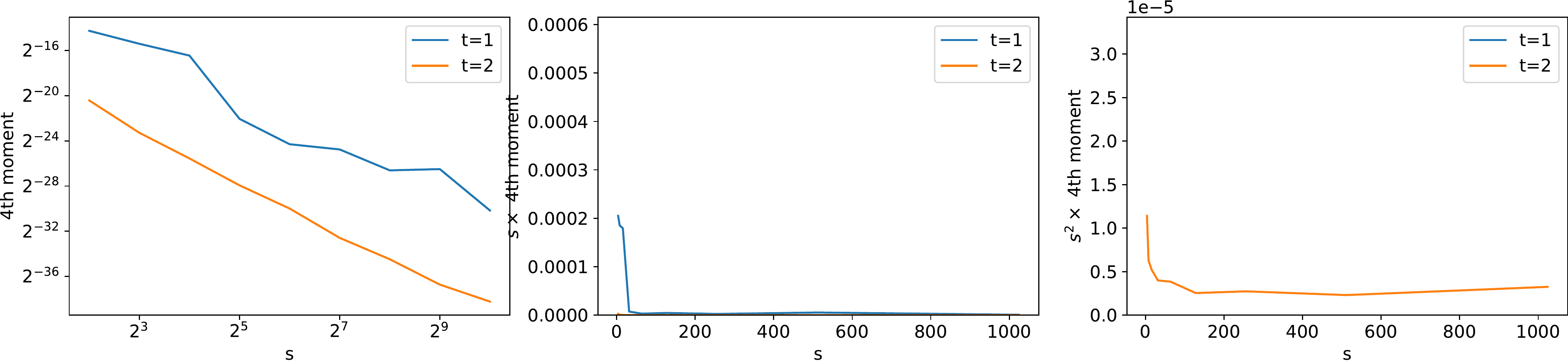}
  \caption{4th moment experiments on the Disjoint 64 Non-Zeros data set.}
  \label{fig:4th64Disj}
\end{figure*}
\begin{figure*}[h!t]
  \centering
  \includegraphics[width=0.93\textwidth]{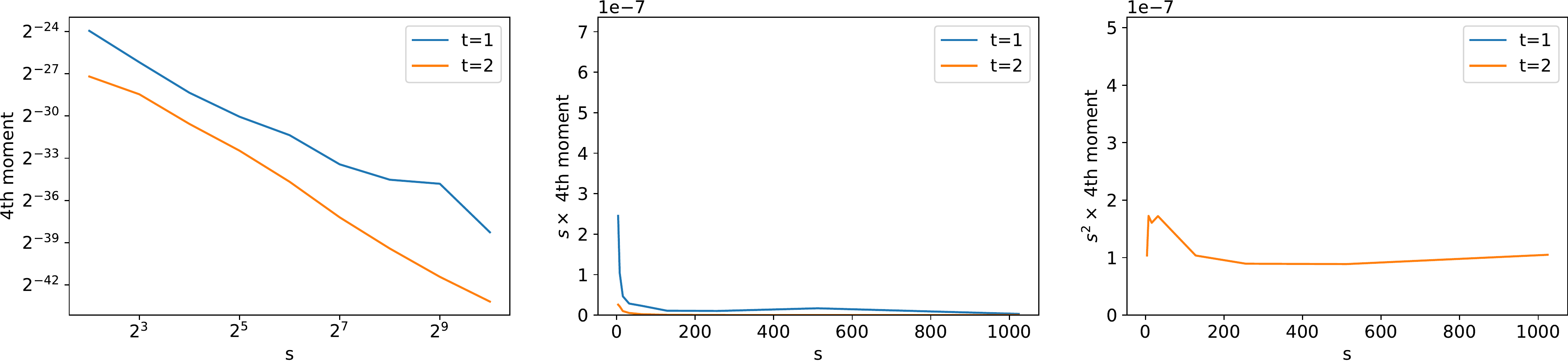}
  \caption{4th moment experiments on the News20 data set.}\label{fig:4thNews20}
\end{figure*}
}

\short{
\section{Conclusion}

We have seen that taking the median of 3 estimates can significantly improve the accuracy of estimates for Count\-Sketch. An interesting direction, that we leave open, is to take advantage of this in more applications that use CountSketch or feature hashing. 
A challenge is that a median operation not available in some contexts (like neural networks, or kernel approximation), and may need to be replaced by a continuously differentiable approximation.

\newpage
}

\raggedbottom

\bibliography{bibliography.bib}
\bibliographystyle{plainnat}

\end{document}